\documentclass[10pt,a4paper]{article}

\usepackage[T1]{fontenc}
\usepackage[utf8]{inputenc}
\usepackage{libertinus}
\usepackage{libertinust1math}
\usepackage[english]{babel}
\usepackage[a4paper,margin=2.15cm,headheight=14pt]{geometry}
\usepackage{microtype}
\usepackage{amsmath,amssymb,mathtools}
\usepackage{amsthm}
\usepackage{aliascnt}
\usepackage{booktabs,tabularx,array,longtable,multirow}
\usepackage{enumitem}
\usepackage{graphicx}
\usepackage{xcolor}
\usepackage[normalem]{ulem}
\usepackage{fancyhdr}
\usepackage{hyperref}
\usepackage[nameinlink,noabbrev,capitalise]{cleveref}
\usepackage{caption}
\usepackage{float}
\usepackage{url}

\hypersetup{
  colorlinks=true,
  linkcolor=black,
  citecolor=black,
  urlcolor=blue,
  pdftitle={Bellman Search in Arbitrary Finite Dimension: A Self-Similar Cell Theorem and Effective Computability of Planar Shoreline Search},
  pdfauthor={Florentin Koch},
  pdfsubject={Online geometric search, self-similar cells, and computability},
  pdfkeywords={online search, shoreline search, hyperplane search, self-similarity, support function, computability, semialgebraic geometry, Bellman dynamics}
}

\newcommand{\figureonefile}{Figure1_English_Clean.pdf}
\newtheorem{theorem}{Theorem}[section]
\newaliascnt{lemma}{theorem}
\newtheorem{lemma}[lemma]{Lemma}
\aliascntresetthe{lemma}
\newaliascnt{proposition}{theorem}
\newtheorem{proposition}[proposition]{Proposition}
\aliascntresetthe{proposition}
\newaliascnt{corollary}{theorem}
\newtheorem{corollary}[corollary]{Corollary}
\aliascntresetthe{corollary}
\theoremstyle{definition}
\newaliascnt{definition}{theorem}

\aliascntresetthe{definition}
\theoremstyle{remark}
\newaliascnt{remark}{theorem}
\newtheorem{remark}[remark]{Remark}
\aliascntresetthe{remark}

\crefname{theorem}{theorem}{theorems}
\Crefname{theorem}{Theorem}{Theorems}
\crefname{lemma}{lemma}{lemmas}
\Crefname{lemma}{Lemma}{Lemmas}
\crefname{proposition}{proposition}{propositions}
\Crefname{proposition}{Proposition}{Propositions}
\crefname{corollary}{corollary}{corollaries}
\Crefname{corollary}{Corollary}{Corollaries}
\crefname{remark}{remark}{remarks}
\Crefname{remark}{Remark}{Remarks}
\crefname{section}{Section}{Sections}

\DeclareMathOperator{\conv}{conv}
\DeclareMathOperator{\inrad}{inrad}
\DeclareMathOperator{\CR}{CR}
\DeclareMathOperator{\Per}{Per}
\DeclareMathOperator{\Ext}{Ext}

\newcommand{\R}{\mathbb{R}}
\newcommand{\Sph}{\mathbb{S}}

\newcommand{\norm}[1]{\lVert #1\rVert}
\newcommand{\abs}[1]{\lvert #1\rvert}
\newcommand{\ip}[2]{\langle #1,#2\rangle}
\newcommand{\cellratio}{\mathcal{R}}
\newcommand{\vmax}{\mathbin{\vee}}
\newcommand{\eps}{\varepsilon}
\newcommand{\spc}{C_{\mathrm{sp}}}

\setlist[itemize]{leftmargin=1.5em,itemsep=0.2em,topsep=0.25em}
\setlist[enumerate]{leftmargin=1.75em,itemsep=0.25em,topsep=0.25em}
\title{\textbf{Bellman Search in Arbitrary Finite Dimension:}\\
A Self-Similar Cell Theorem and Effective Computability of Planar Shoreline Search}
\author{Florentin Koch\\[0.3em]
\normalsize \'Ecole Polytechnique\\
\normalsize \texttt{florentin.koch@polytechnique.edu}}
\date{29 August 2026 -- English revision 8.3}

\begin{document}
\maketitle

\begin{abstract}
A shoreline-search path starts at the origin and must meet an unknown affine line, without knowing either its normal or its distance. We first establish a self-similar reduction theorem for homogeneous search problems whose historical information is a record profile updated by pointwise maximum. Two quasi-returns of the normalized state delimit a block that renews the required profile by itself; a short connector closes this block into a cell. Every finite-ratio path can therefore be approximated, with arbitrarily small loss, by repetitions of a single cell at all scales.

The main chain is then made effective. A finite coding of the state space computably bounds the scale factor and normalized length of a nearly optimal cell. For planar Shoreline search, the support function of the convex hull gives an exact cell functional. A one-sided polygonalization then reduces the problem to a computable number of vertices, after which quantifier elimination decides whether a polygonal cell exists below a rational threshold.

It follows that the optimal deterministic planar Shoreline value $C_2^*$ is a computable real: for every rational $\eps>0$, an algorithm terminates with a rational interval of width at most $\eps$ containing $C_2^*$. Additional results---sliding memory, Bellman transitions, deadlines, geometric filters, and relative equilibria---are presented separately as a toolbox for certified computation and for the study of spiral rigidity; they are not used in the computability proof.
\end{abstract}


\section{Introduction}
\subsection{The problem and the main result}
A Shoreline strategy is a locally rectifiable curve
\[
  \gamma:[0,\infty)\longrightarrow \R^2,\qquad \gamma(0)=0,
\]
parametrized by arc length. The adversary chooses a unit normal $u\in\Sph^1$ and a distance $D>0$, hence the line
\[
  H(u,D):=\{x\in\R^2:\ip{u}{x}=D\}.
\]
If $\tau_\gamma(u,D)$ is the first time at which $\gamma$ meets this line, the purely multiplicative competitive ratio is
\[
  \CR(\gamma):=\sup_{u\in\Sph^1,\,D>0}\frac{\tau_\gamma(u,D)}{D},
  \qquad
  C_2^*:=\inf_\gamma \CR(\gamma).
\]
The best explicit path currently known is a logarithmic spiral with ratio
\[
  \spc\simeq 13.81113517946,
\]
whereas Temerev's unconditional deterministic lower bound gives
\[
  12.5937096701246675\ldots\le C_2^*\le \spc
\]
\cite{FinchZhu,FinchSpiral,Temerev}.

The present paper proves that $C_2^*$ is effectively computable. This gives a certifiable alternative: if $C_2^*<\spc$, finite precision will eventually produce a certified interval lying strictly below $\spc$; if $C_2^*=\spc$, the certified intervals will converge to the spiral value. Determining which scenario occurs remains a separate structural problem.

A companion paper extends the effective-computability theorem to every fixed finite dimension. It also develops a conjectural unifying mechanism in which the one-dimensional zigzag, the planar logarithmic spiral, and higher-dimensional precessing exponential oscillators arise from the same scale-normalized Bellman dynamics with max-memory. Those later structural conjectures are independent of, and are not used in, the present proof.

\subsection{Lineage and previous work}
These problems belong to the tradition of geometric minimax search initiated by Bellman: Gross attributed a search problem to him in 1955, and Bellman subsequently related it to dynamic programming and geometric escape \cite{Gross,BellmanMin,BellmanDP}. The directly relevant milestones are summarized in \cref{tab:history}; each treats a particular information model or class of paths, whereas the present theorem extracts a repeated cell from an otherwise unrestricted path.

\begin{table}[ht]
\centering
\small
\caption{Directly relevant milestones and the limits of their scope.}
\label{tab:history}
\begin{tabularx}{\textwidth}{@{}p{1.5cm}p{2.6cm}X X@{}}
\toprule
Year & Reference & Contribution & Not covered there \\
\midrule
1955--57 & Gross; Bellman \cite{Gross,BellmanMin,BellmanDP} & Origins of geometric minimax search and the forest-escape problem. & Arbitrary unknown lines and a value-preserving self-similar reduction. \\
1957 & Isbell \cite{Isbell} & Search for a line when its scale is known; optimal constant for the more informed problem. & Unknown distance and scale invariance. \\
1964--70 & Beck; Beck--Newman \cite{Beck,BeckNewman} & Linear search and the optimal zigzag ratio $9$. & Continuous coupling of all planar directions. \\
1993 & Baeza--Yates et al. \cite{BaezaYates} & Systematic formulation of planar search with partial information. & Global optimality of the logarithmic spiral. \\
2005 & Finch--Zhu; Finch \cite{FinchZhu,FinchSpiral} & Shoreline formulation and the logarithmic-spiral conjecture. & Reduction of a free path to one repeated cell. \\
2010--12 & Langetepe \cite{LangetepeSODA,LangetepeAxis} & Spirality or cyclicity in structured variants. & Unrestricted Shoreline search with a continuum of normals. \\
2020 & Dobrev et al.; Acharjee et al. \cite{Dobrev,Acharjee} & Improved bounds and multi-robot variants. & Exact value of the single-agent problem. \\
2022--23 & Antoniadis et al.; Bansal et al. \cite{Antoniadis,Bansal} & Dimensional bounds for hyperplane search. & Optimal constants and classification of states. \\
2024 & Ghomi--Wenk \cite{GhomiWenk} & Shortest closed curves whose convex hull contains a Euclidean sphere. & The online ordering, first-passage constraint, and max-memory dynamics of Shoreline search. \\
2026 & Temerev \cite{Temerev} & Unconditional planar deterministic lower bound $12.5937096701\ldots$. & The gap to the reference spiral value $13.811135\ldots$. \\
\bottomrule
\end{tabularx}
\end{table}

The closest conceptual neighbor of our reduction is the multiscale construction of Antoniadis et al. \cite{Antoniadis}, in which a preselected inspection curve is repeated at several scales. Here, by contrast, the cell is extracted from an arbitrary nearly optimal path and preserves, up to $\eps$, the genuine online objective on every prefix. The convex-hull viewpoint of Ghomi--Wenk is also geometrically close to the terminal inradius bounds used below, but it does not encode online order or the synchronized prefix objective.

\subsection{Contributions and logical status of the results}
The paper has three main contributions.
\begin{description}[leftmargin=3em,labelwidth=2.4em,style=nextline]
\item[C1.] A value-preserving self-similar reduction: every admissible finite-ratio history can be replaced, with arbitrarily small loss, by one cell repeated at all scales.
\item[C2.] An effectivization of this reduction: at fixed precision, the scale factor, normalized length, and ultimately the number of vertices of a nearly optimal cell are bounded by computable functions.
\item[C3.] For Shoreline search, an exact translation through support functions and centered inradii, followed by a semialgebraic reduction that makes $C_2^*$ computable.
\end{description}
\Cref{tab:logic} separates the main proof chain from auxiliary results.

\begin{table}[ht]
\centering
\small
\caption{Logical role of the results in the paper.}
\label{tab:logic}
\begin{tabularx}{\textwidth}{@{}p{2.7cm}X X@{}}
\toprule
Category & Content & Role \\
\midrule
Main chain & \Cref{thm:cell}, \cref{thm:effective-return}, \cref{prop:support-certificate,prop:cell-functional}, \cref{thm:polygonalization}, \cref{prop:qe}, and \cref{thm:computability}. & Each link is used to establish the next. \\
Structural scope & \Cref{cor:equivariant}, \cref{thm:metric-cone}, and \cref{cor:hyperplane}. & Shows that the cell reduction is not specific to planar Shoreline search. \\
Toolbox & \Cref{thm:recut}, \cref{thm:sliding}, \cref{thm:bellman-step}, \cref{prop:deadline}, \cref{thm:relative-equilibrium}, and the filters of \cref{sec:filters}. & Supports certified computation and rigidity analysis; not used in \cref{thm:computability}. \\
Numerical, not certified & Numerical values that are not accompanied by a stand-alone certificate are explicitly labelled as such and are not used in the main computability theorem. \\
\bottomrule
\end{tabularx}
\end{table}

\subsection{Overview of the proof}
Before introducing the formalism, here is the complete argument in five links.

\paragraph{Link 1: from an infinite history to one cell.}
An admissible path is an infinite-dimensional object. At two widely separated scales, a quasi-return of the normalized state, together with max-memory, forces the intermediate block to renew the dilated old profile by itself. A short connector closes this block into a single cell whose repetition at all scales loses at most $\eps$; see \cref{thm:cell}.

\paragraph{Link 2: locating and bounding the cell.}
The first link does not say where to find the quasi-return or how long the cell may be. We replace abstract precompactness by an explicit finite coding of normalized states and apply a counted pigeonhole principle. This gives computable bounds on the return index, scale factor, and normalized length; see \cref{thm:effective-return}.

\paragraph{Link 3: from a continuous curve to a finite polygon.}
Even after it is bounded, a cell remains an infinite-dimensional curve. We first express the ratio through support functions and a centered inradius, then replace each short arc by its chord while comparing it with the time at which the whole original arc is already available. The result is a polygonal cell with a computable number of vertices and controlled loss; see \cref{prop:support-certificate,prop:cell-functional} and \cref{thm:polygonalization}.

\paragraph{Link 4: deciding a rational threshold.}
The vertex coordinates are still continuous real variables. The assertion ``there exists a polygonal cell of ratio at most $c$'' is written as a first-order formula over the reals, and quantifier elimination gives a terminating yes/no test; see \cref{prop:qe}.

\paragraph{Link 5: computing the value.}
A threshold oracle is queried by rational bisection. This returns a certified interval of prescribed width containing $C_2^*$; see \cref{thm:computability}.

In schematic form,
\[
\begin{aligned}
\text{infinite history}
&\xrightarrow{\text{quasi-return + max-memory}}
\text{cell}
\xrightarrow{\text{finite coding}}
\text{bounded cell},\\
\text{bounded cell}
&\xrightarrow{\text{support/inradius + chords}}
\text{polygon with }N(\eps)\text{ vertices}
\xrightarrow{\text{quantifier elimination}}
\text{threshold test}
\xrightarrow{\text{bisection}}
\eps\text{-interval}.
\end{aligned}
\]
The results in Part~II---sliding memory, Bellman transitions, deadlines, filters, and relative equilibria---do not enter any of these five links. They are presented separately as tools for future implementations and for structural analysis of the spiral.

\part{The main proof chain}

\section{Homogeneous framework and max-memory}
\subsection{Scale, constraints, and service}
The next two sections introduce only the data needed for the cell theorem: a scaling action, a record profile, and a normalized state. This is not a theory parallel to Shoreline search; it isolates the common mechanism that will later be specialized to hyperplane search.

Let $(X,d,o)$ be a pointed metric space. For each $q>0$, let $S_q$ denote scaling by the factor $q$. We assume
\begin{equation}
S_1=\mathrm{Id},\qquad S_qS_r=S_{qr},\qquad S_q(o)=o,\qquad d(S_qx,S_qy)=q\,d(x,y).
\label{eq:scaling}
\end{equation}
In a normed space, $o=0$ and $S_qx=qx$.

A compact metric space $U$ parametrizes the adversarial constraints. A continuous function
\[
\phi:U\times X\longrightarrow\R
\]
measures instantaneous service: $\phi(u,x)$ is the level supplied to constraint $u$ at position $x$. It is homogeneous of degree one when
\begin{equation}
\phi(u,S_qx)=q\phi(u,x).
\label{eq:service-homogeneity}
\end{equation}
A finite history is a rectifiable curve; an infinite history is locally rectifiable. The stability properties required of the admissible class will be stated together with the cell theorem, where they are actually used.

\subsection{Records, completion, and ratio}
For an arc-length parametrized history $\gamma$, define the record of constraint $u$ at time $t$ by
\begin{equation}
R_\gamma(t,u):=\max_{0\le s\le t}\phi(u,\gamma(s)),
\qquad
m_\gamma(t):=\min_{u\in U}R_\gamma(t,u).
\label{eq:records}
\end{equation}
The level $m_\gamma(t)$ is the amount guaranteed simultaneously to all constraints. For $d>0$, define the first completion time
\begin{equation}
T_\gamma(d):=\inf\{t\ge0:m_\gamma(t)\ge d\},\qquad \inf\varnothing:=+\infty,
\label{eq:completion}
\end{equation}
and the competitive ratio
\begin{equation}
\CR(\gamma):=\sup_{d>0}\frac{T_\gamma(d)}{d}.
\label{eq:cr-general}
\end{equation}
This purely multiplicative convention applies at every scale $d>0$. It differs from models that impose a minimum target distance or allow an additive constant; neither variant is used in the proofs.

\begin{lemma}[Phase form of the ratio]
\label{lem:phase-ratio}
If $m_\gamma$ is continuous, nondecreasing, and unbounded, then
\begin{equation}
\CR(\gamma)=\sup_{t:m_\gamma(t)>0}\frac{t}{m_\gamma(t)}.
\label{eq:phase-ratio}
\end{equation}
In particular, $\CR(\gamma)\le C<\infty$ implies $m_\gamma(t)\ge t/C$ and $m_\gamma(t)\to\infty$.
\end{lemma}

\begin{proof}
Continuity gives $m_\gamma(T_\gamma(d))=d$ for every attained level, so the right-hand side of \eqref{eq:phase-ratio} is at least \eqref{eq:cr-general}. Conversely, if $m_\gamma(t)>0$, choose $d_n\downarrow m_\gamma(t)$ with $d_n>m_\gamma(t)$. Then $T_\gamma(d_n)\ge t$, hence
\[
\sup_{d>0}\frac{T_\gamma(d)}{d}
\ge \lim_{n\to\infty}\frac{t}{d_n}
=\frac{t}{m_\gamma(t)}.
\]
The final assertion follows immediately from \eqref{eq:phase-ratio}.
\end{proof}

\subsection{Max-memory}
If two compatible pieces are concatenated, their terminal record profile is the pointwise maximum of the two profiles:
\begin{equation}
R_{\gamma_1\star\gamma_2}(u)
=R_{\gamma_1}(u)\vmax R_{\gamma_2}(u),
\qquad
(F\vmax G)(u):=\max\{F(u),G(u)\}.
\label{eq:max-memory}
\end{equation}
This law is idempotent: reproducing an already acquired record does not change the state. It also supplies the renewal mechanism used below: if the final profile is strictly above the old profile in every direction, the increase must have been produced by the new block.

\section{Normalization and quasi-returns}
This section prepares a single event: two very distant scales at which cost, position, and historical profile are nearly identical after normalization.

\subsection{States at geometric levels}
Fix $\lambda>1$ and levels $r_n=\lambda^n$. By \cref{lem:phase-ratio}, the times
\[
T_n:=T_\gamma(r_n)
\]
are finite whenever $\CR(\gamma)<\infty$. In units of the scale $r_n$, set
\begin{equation}
\ell_n:=\frac{T_n}{r_n},
\qquad
x_n:=S_{1/r_n}\gamma(T_n),
\qquad
H_n(u):=\frac{R_\gamma(T_n,u)}{r_n}.
\label{eq:normalized-state-components}
\end{equation}
The normalized state is
\begin{equation}
\Sigma_n:=(\ell_n,x_n,H_n),
\qquad
\min_{u\in U}H_n(u)=1.
\label{eq:normalized-state}
\end{equation}
It retains cost, position, and all records after quotienting out the absolute scale.

\subsection{Precompactness and quasi-returns}
Equip the state space with the metric
\begin{equation}
D_\Sigma\bigl((\ell,x,H),(\ell',x',H')\bigr)
:=\abs{\ell-\ell'}+d(x,x')+\norm{H-H'}_\infty.
\label{eq:state-metric}
\end{equation}
The sequence $(\Sigma_n)$ is called precompact when its closure is compact. Equivalently, the normalized states cannot contain infinitely many configurations separated by a fixed positive distance; arbitrarily accurate quasi-returns must occur.

\begin{lemma}[Recurrence with arbitrarily large gap]
\label{lem:large-gap-return}
If $(\Sigma_n)$ is precompact, then for every $\delta>0$ and $N_0\in\mathbb{N}$ there exist $i<j$ such that
\[
j-i>N_0,
\qquad
D_\Sigma(\Sigma_i,\Sigma_j)<\delta.
\]
\end{lemma}

\begin{proof}
A precompact sequence has a convergent subsequence. Two sufficiently late terms of that subsequence are within $\delta$ and can be chosen with index gap greater than $N_0$.
\end{proof}

The associated scale ratio is
\begin{equation}
q:=\frac{r_j}{r_i}=\lambda^{j-i};
\label{eq:q-ratio}
\end{equation}
it can therefore be made arbitrarily large while preserving an accurate normalized quasi-return.

\section{Self-similar cell theorem}
\subsection{Structural hypotheses}
We consider a class of admissible histories satisfying the following properties.

\begin{description}[leftmargin=3.6em,labelwidth=3em]
\item[(H0)] The class is stable under restriction, finite concatenation of compatible pieces, scaling, and monotone arc-length reparametrization. It is also stable under the all-scale concatenation $\star_{k\in\mathbb Z}S_{Q^k}C$ whenever the negative-index copies have finite total length and their endpoints converge to $o$.
\item[(H1)] Records obey the max-law \eqref{eq:max-memory} and the scale covariance induced by \eqref{eq:service-homogeneity}.
\item[(H2)] For every finite-ratio history, the normalized states \eqref{eq:normalized-state} are precompact.
\item[(H3)] On every sublevel $\CR\le C$, there is $K_C<\infty$ such that a quasi-return of error $\delta$ between levels $r_i$ and $r_j=qr_i$ can be closed by an admissible connector of cost at most $K_Cqr_i\delta$.
\end{description}
Hypothesis (H0) is only closure under the operations actually used by the construction. For Shoreline search, where strategies are locally rectifiable curves starting at the origin, it is automatic. The remaining hypotheses are geometric and verifiable; \cref{sec:abstract-scope} proves them in finite-dimensional normed spaces.

\subsection{Statement and proof}
For a rectifiable cell $C$ joining $x$ to $S_Qx$, with $Q>1$, define its repetition at all scales by the ordered concatenation
\[
\Gamma_C:=\mathop{\star}_{k\in\mathbb Z}S_{Q^k}C.
\]
The negative-index copies contract toward $o$, while the positive-index copies explore arbitrarily large scales. Whenever this repetition is admissible, write
\[
\cellratio_Q(C):=\CR(\Gamma_C).
\]

\begin{theorem}[Homothetic self-similar reduction]
\label{thm:cell}
Assume \emph{(H0)--(H3)}. Let $\gamma$ be an admissible history with $\CR(\gamma)\le C<\infty$. For every $\eps>0$, there exist $Q>1$, a point $x\in X$, a length $L<\infty$, and a rectifiable curve $C_0:[0,L]\to X$ satisfying
\[
C_0(0)=x,
\qquad
C_0(L)=S_Qx,
\]
such that $\Gamma_{C_0}$ is a locally rectifiable admissible history starting at $o$ and
\[
\CR(\Gamma_{C_0})\le C+\eps.
\]
Consequently, the infimum of the ratio over all admissible histories equals the infimum over repetitions of homothetic cells.
\end{theorem}

\begin{quote}\small
\textit{Interpretation.} A long-gap quasi-return makes the normalized states almost agree. Max-memory then forces the intervening block to recreate the required dilated profile by itself. A short connector makes the quasi-return exact, and the geometric past has finite cost $L/(Q-1)$. The ratio is controlled at every internal phase, not only at cell endpoints.
\end{quote}

\begin{proof}
Fix $\lambda>1$. By \cref{lem:large-gap-return}, choose $i<j$ such that
\begin{equation}
\abs{\ell_j-\ell_i}<\delta,
\qquad
d(x_j,x_i)<\delta,
\qquad
\norm{H_j-H_i}_\infty<\delta,
\label{eq:quasi-return-data}
\end{equation}
for a parameter $\delta\in(0,1/2)$ to be chosen sufficiently small, and with $j-i$ sufficiently large. Put
\begin{equation}
q:=\frac{r_j}{r_i},
\qquad
Q:=q(1-\delta)>1.
\label{eq:q-Q}
\end{equation}

\emph{1. Domination at scale $Q$.}
Since $H_i\ge1$ pointwise,
\[
H_j\ge H_i-\delta\ge(1-\delta)H_i,
\qquad
qH_j\ge QH_i.
\]
In physical units,
\begin{equation}
R_\gamma(T_j,\cdot)\ge Q R_\gamma(T_i,\cdot).
\label{eq:profile-domination}
\end{equation}

\emph{2. Renewal by the new block.}
Let $B_{i,t}$ be the profile produced by the block $\gamma([T_i,t])$ alone. The max-law gives
\[
R_\gamma(t,\cdot)=R_\gamma(T_i,\cdot)\vmax B_{i,t}.
\]
At $t=T_j$, the target profile in \eqref{eq:profile-domination} is strictly above $R_\gamma(T_i,\cdot)$ because $Q>1$; therefore the new block must satisfy
\begin{equation}
B_{i,T_j}\ge Q R_\gamma(T_i,\cdot).
\label{eq:block-renewal}
\end{equation}
In fact, for every $u\in U$,
\[
R_\gamma(T_j,u)\ge Q R_\gamma(T_i,u)>R_\gamma(T_i,u),
\]
and because
\[
R_\gamma(T_j,u)=\max\{R_\gamma(T_i,u),B_{i,T_j}(u)\},
\]
the maximum cannot be supplied by the old profile. Hence
\begin{equation}
B_{i,T_j}=R_\gamma(T_j,\cdot),
\qquad
\min_{u\in U}B_{i,T_j}(u)=r_j=qr_i.
\label{eq:exact-block-profile}
\end{equation}
Thus the intermediate block recreates the entire required terminal profile by itself; in particular, level $r_j$ is already acquired before the connector is traversed.

\emph{3. Closing the quasi-return.}
By (H3), connect $\gamma(T_j)$ to $S_Q\gamma(T_i)$ by an admissible connector of length
\begin{equation}
e\le K_Cqr_i\delta.
\label{eq:connector-cost}
\end{equation}
Let $C_0$ be the block $\gamma([T_i,T_j])$ followed by this connector. Its length is
\[
L=T_j-T_i+e.
\]

\emph{4. Cost of the geometric past.}
Before a given copy, all earlier copies have total length $L/(Q-1)$. Since $T_i=r_i\ell_i$, $T_j=qr_i\ell_j$, and $Q=q(1-\delta)$,
\[
T_j-QT_i
=qr_i(\ell_j-\ell_i)+\delta T_i
\le qr_i\delta(1+C).
\]
Consequently,
\begin{equation}
\frac{L}{Q-1}-T_i
\le \eta_C(\delta,q)r_i,
\qquad
\eta_C(\delta,q):=
\frac{q\delta(1+C+K_C)}{q(1-\delta)-1}.
\label{eq:eta}
\end{equation}
As $\delta\to0$ and $q\to\infty$, $\eta_C(\delta,q)\to0$.

\emph{5. Control of every phase.}
Consider, in the copy $S_{Q^k}C_0$, the phase corresponding to $t\in[T_i,T_j]$. The immediately preceding copy supplies at least $Q^kR_\gamma(T_i,\cdot)$, while the current block supplies $Q^kB_{i,t}$. The complete historical profile therefore dominates
\[
Q^k\bigl(R_\gamma(T_i,\cdot)\vmax B_{i,t}\bigr)
=Q^kR_\gamma(t,\cdot).
\]
Older copies can only improve this lower bound. The guaranteed level is therefore at least $Q^km_\gamma(t)$, while historical cost is at most
\[
Q^k\left(\frac{L}{Q-1}+t-T_i\right)
\le Q^k\bigl(t+\eta_C(\delta,q)r_i\bigr).
\]
Since $m_\gamma(t)\ge r_i$ and $t/m_\gamma(t)\le C$,
\[
\frac{\text{cost}}{\text{guaranteed level}}
\le C+\eta_C(\delta,q).
\]
During the connector, the already-created level is at least $r_j=qr_i$, while the extra relative cost is bounded by $\eta_C(\delta,q)/q+K_C\delta$. Choose first $\delta$ small and then $j-i$ large so that both errors are below $\eps$.

Finally,
\[
\sum_{k<0}Q^kL=\frac{L}{Q-1}<\infty,
\]
and the endpoints of the negative-index copies converge to $o$. By (H0), their arc-length concatenation is an admissible history starting at $o$. The reverse inequality between the two infima is immediate because repeated cells form a subclass of all admissible histories.
\end{proof}

\begin{remark}
\label{rem:cell-shape}
\Cref{thm:cell} is a reduction of value. The extracted cell may still contain returns, corners, and changes of active provider; no spiral or polygonal form is imposed.
\end{remark}

\section{Effective return and a uniformly bounded cell}
\Cref{thm:cell} produces a finite-length cell, but it does not bound the quasi-return index. Computability requires a uniform version: at prescribed precision, a sufficiently separated return must occur before a computable rank.

\begin{theorem}[Effectively bounded self-similar return in $\R^D$]
\label{thm:effective-return}
Fix an integer $D\ge1$ and rational numbers $U<\infty$, $\zeta>0$, and $\lambda>1$. There exists a computable integer
\[
N_{\mathrm{ret}}=N_{\mathrm{ret}}(D,U,\zeta,\lambda)
\]
such that every hyperplane-search path $\gamma$ in $\R^D$ with $\CR(\gamma)\le U$ admits a cell $C_0$, extracted between two levels $r_i=\lambda^i$ and $r_j=\lambda^j$, for which
\[
\cellratio_Q(C_0)\le \CR(\gamma)+\zeta,
\qquad
1\le j-i\le N_{\mathrm{ret}}.
\]
The scale factor and the length of $C_0$ in units of $r_i$ are therefore bounded by computable functions of $(D,U,\zeta,\lambda)$.
\end{theorem}

\begin{quote}\small
\textit{Interpretation.} Abstract precompactness is replaced by an explicit finite code for normalized states. A counted pigeonhole argument then yields computable bounds on the return rank, the scale factor, and the normalized cell length.
\end{quote}

\begin{proof}
Set
\begin{equation}
\delta:=\min\left\{\frac18,\frac{\zeta}{128(1+U)}\right\}.
\label{eq:effective-delta}
\end{equation}
For $\CR(\gamma)\le U$, normalized states satisfy
\begin{equation}
1\le \ell_n\le U,
\qquad
\norm{x_n}\le U,
\qquad
1\le H_n(u)\le U,
\qquad
\abs{H_n(u)-H_n(v)}\le U\norm{u-v}.
\label{eq:state-bounds}
\end{equation}
Appendix~\ref{app:covering} constructs an explicit computable number $M(D,U,\delta)$ of codes such that two states with the same code are at distance less than $\delta$.

Choose
\[
L_0:=\left\lceil\frac{\log4}{\log\lambda}\right\rceil,
\qquad
\lambda^{L_0}\ge4,
\]
and consider the $M+1$ states
\[
\Sigma_0,\Sigma_{L_0},\ldots,\Sigma_{ML_0}.
\]
Two share a code. For the corresponding indices $i<j$,
\[
L_0\le j-i\le ML_0,
\qquad
4\le q=\lambda^{j-i}\le\lambda^{ML_0}.
\]
In $\R^D$, the straight connector satisfies
\[
e\le qr_i\delta(1+U).
\]
The proof of \cref{thm:cell} then gives
\[
\eta\le
\frac{2q\delta(1+U)}{q(1-\delta)-1}
\le8\delta(1+U)
\le\frac{\zeta}{16},
\]
and the connector adds at most $\eta/q+\delta(1+U)<\zeta/16$. The resulting cell therefore has ratio less than $\CR(\gamma)+\zeta$.

One may take
\begin{equation}
N_{\mathrm{ret}}:=M(D,U,\delta)L_0.
\label{eq:Nret}
\end{equation}
The principal block has length at most $Uqr_i$ and the connector at most $qr_i\delta(1+U)$, so
\begin{equation}
\frac{L(C_0)}{r_i}
\le \lambda^{N_{\mathrm{ret}}}\bigl(U+\delta(1+U)\bigr)
=:L_{\max}(D,U,\zeta,\lambda).
\label{eq:Lmax}
\end{equation}
Every displayed quantity is computable.
\end{proof}

\begin{remark}
\Cref{thm:effective-return} bounds the scale span and length of a nearly optimal cell. It does not yet bound the number of turns; polygonalization in \cref{sec:polygonalization} removes this final infinitude.
\end{remark}

\section{Abstract scope of the reduction}
\label{sec:abstract-scope}
This section adds no link to the computability proof. It shows that the mechanism of \cref{thm:cell} is stable under compact symmetries and applies to proper metric cones, hence in particular to finite-dimensional normed spaces.

\subsection{Compact symmetries}
Let $G$ be a compact group of isometries of $X$. Assume that $G$ fixes $o$, preserves admissibility and cost, commutes with scaling,
\[
gS_q=S_qg,
\]
and acts on $U$ so that
\begin{equation}
\phi(u,S_qgx)=q\phi(g^{-1}u,x).
\label{eq:equivariant-service}
\end{equation}
The action on profiles is $(g\cdot H)(u):=H(g^{-1}u)$.

\begin{corollary}[Equivariant reduction]
\label{cor:equivariant}
Suppose that, on every sublevel $\CR\le C$, normalized profiles satisfy $\norm{H}_\infty\le B_C$ and normalized states are precompact modulo $G$. Then every history of ratio at most $C$ and every $\eps>0$ admit $Q>1$, $g\in G$, and a cell $C_0:x\to S_Qgx$ whose repetition under $(S_Qg)^k$ has ratio at most $C+\eps$.
\end{corollary}

\begin{proof}
Choose a quasi-return for which $\Sigma_j$ is $\delta$-close to $g\Sigma_i$ and whose gap is large enough that
\[
Q:=q(1-\delta)>B_C+1.
\]
Then
\[
qH_j\ge Q(g\cdot H_i),
\qquad
H_i<Q(g\cdot H_i)
\]
pointwise, because $\min(g\cdot H_i)=1$ and $H_i\le B_C$. The max-law forces the new block to recreate the transformed profile $Q(g\cdot H_i)$ by itself. The connector closes the block by joining $\gamma(T_j)$ to $S_Qg\gamma(T_i)$. The cost estimates of \cref{thm:cell} are unchanged, and commutation of $g$ with $S_q$ ensures compatibility of successive copies.
\end{proof}

In the plane, $G=\mathrm{SO}(2)$ gives direct similarities $QR_\Delta$ and the relative equilibria of \cref{sec:relative-equilibria}.

\subsection{Proper metric cones}
A proper metric cone is a proper geodesic space $(X,d,o)$ equipped with the scalings \eqref{eq:scaling} and satisfying
\begin{equation}
d(S_qx,S_rx)=\abs{q-r}\,d(o,x).
\label{eq:metric-cone}
\end{equation}

\begin{theorem}[Geometric verification of the hypotheses]
\label{thm:metric-cone}
Let $(X,d,o)$ be a proper metric cone, let $U$ be compact, and let $\phi:U\times X\to\R$ be continuous, homogeneous of degree one, with $\phi(u,o)=0$ and
\[
\abs{\phi(u,x)-\phi(u,y)}\le L_\phi d(x,y).
\]
Assume that the class of histories satisfies \emph{(H0)} and that geodesic segments are admissible pieces. Then every finite-ratio history satisfies \emph{(H1)--(H3)}. Normalized profiles are also uniformly bounded on ratio sublevels.
\end{theorem}

\begin{proof}
If a finite-ratio history exists, then necessarily $L_\phi>0$. If $\CR(\gamma)\le C$, then at the times $T_n$,
\[
\frac{1}{L_\phi}\le\ell_n\le C,
\qquad
d(o,x_n)\le C,
\qquad
H_n(u)\le L_\phi C.
\]
Indeed, for a constraint at which level $r_n$ is reached at time $T_n$, the inequality $\abs{\phi(u,\gamma(s))}\le L_\phi s$ gives $r_n\le L_\phi T_n$. Positions stay in the compact ball $B(o,C)$. Uniform continuity of $\phi$ on $U\times B(o,C)$ gives equicontinuity of the profiles, and Arzel\`a--Ascoli yields their precompactness.

If $x_i$ and $x_j$ are $\delta$-close after normalization, with $q=r_j/r_i$ and $Q=q(1-\delta)$, then
\begin{align*}
d\bigl(\gamma(T_j),S_Q\gamma(T_i)\bigr)
&\le d(S_{qr_i}x_j,S_{qr_i}x_i)
   +d(S_{qr_i}x_i,S_{Qr_i}x_i)\\
&\le qr_i\delta+(q-Q)r_i d(o,x_i)\\
&\le qr_i\delta(1+C).
\end{align*}
The geodesic segment supplies the required admissible connector.
\end{proof}

\subsection{Normed spaces and the dual sphere}
Let $X$ be a finite-dimensional normed space and $X^*$ its dual. Its dual unit sphere is
\[
\Sph_{X^*}:=\{f\in X^*:\norm{f}_*=1\}.
\]
Every $f\in\Sph_{X^*}$ and $D>0$ define the affine hyperplane $\{x:f(x)=D\}$. With
\[
\phi(f,x):=f(x),
\]
one has $\abs{f(x)-f(y)}\le\norm{x-y}$ and $\phi(f,qx)=q\phi(f,x)$.

\begin{corollary}[Hyperplane search]
\label{cor:hyperplane}
In every finite-dimensional normed space, the value of online search for hyperplanes whose normals belong to a compact set $U\subset\Sph_{X^*}$ is, up to arbitrarily small error, attained by repetition of a homothetic cell.
\end{corollary}

For the class of all locally rectifiable curves, segments are admissible and (H0) is automatic; \cref{thm:metric-cone} therefore applies without any additional assumption.

\section{Geometric interface for planar Shoreline search}
\subsection{Support function and guaranteed level}
For an arc-length parametrized planar path $\gamma$, set
\begin{equation}
K_t:=\conv\gamma([0,t]),
\qquad
h_t(u):=\max_{0\le s\le t}\ip{u}{\gamma(s)},
\qquad
m(t):=\min_{u\in\Sph^1}h_t(u).
\label{eq:support-level}
\end{equation}
Since $0\in K_t$,
\[
m(t)=\inrad_0(K_t),
\]
the radius of the largest disk centered at the origin and contained in $K_t$ \cite{Schneider}.

\begin{proposition}[Support--inradius certificate]
\label{prop:support-certificate}
The Shoreline ratio is
\begin{equation}
\CR(\gamma)=\sup_{t:m(t)>0}\frac{t}{m(t)}.
\label{eq:shoreline-phase}
\end{equation}
\end{proposition}

\begin{quote}\small
\textit{Interpretation.} The convex hull compiles every directional record, and its centered inradius is exactly the distance already guaranteed in all directions. The adversarial optimization over direction, distance, and phase therefore reduces to the synchronized geometric quotient in \eqref{eq:shoreline-phase}.
\end{quote}

\begin{proof}
At time $t$, for each direction $u$, the path has met every line $H(u,D)$ with $0<D\le h_t(u)$. The common completed level is therefore exactly $m(t)$. The conclusion follows from \cref{lem:phase-ratio}.
\end{proof}

\subsection{Exact functional of a cell}
Let $C_0:x\to Qx$, $Q>1$, be a cell of length $L$. Write $K_C$ for its terminal convex hull and $K_C(t)$ for the hull of its prefix of internal length $t\in[0,L]$.

\begin{lemma}[Nonempty interior for renewing cells]
\label{lem:renewing-interior}
Every Shoreline cell constructed by \cref{thm:cell} satisfies $0\in\operatorname{int}K_C$.
\end{lemma}

\begin{proof}
Renewal \eqref{eq:block-renewal} imposes a uniformly positive support in every direction on the new block. Its convex hull therefore contains a disk centered at the origin.
\end{proof}

Assume henceforth that $0\in\operatorname{int}K_C$. The convex hull of all preceding copies is $Q^{-1}K_C$, because $Q^{-j}K_C\subset Q^{-1}K_C$ for $j\ge1$. The level guaranteed at phase $t$ is
\begin{equation}
m_C(t):=\min_{u\in\Sph^1}
\max\left\{Q^{-1}h_{K_C}(u),h_{K_C(t)}(u)\right\}.
\label{eq:cell-level}
\end{equation}

\begin{proposition}[Exact cell functional]
\label{prop:cell-functional}
The ratio of the all-scale repetition of $C_0$ is
\begin{equation}
\cellratio_Q(C_0):=\CR(\Gamma_{C_0})
=\sup_{0\le t\le L}
\frac{L/(Q-1)+t}{m_C(t)}.
\label{eq:cell-functional}
\end{equation}
\end{proposition}

\begin{quote}\small
\textit{Interpretation.} Before the current cell, all historical cost is $L/(Q-1)$ and all historical support is summarized by the single contracted hull $Q^{-1}K_C$. Formula \eqref{eq:cell-functional} evaluates the true online ratio at every internal phase.
\end{quote}

\begin{proof}
Before the current cell, the total length of older copies is $L/(Q-1)$ and their convex hull is $Q^{-1}K_C$. At phase $t$, the current prefix adds $K_C(t)$. Applying \cref{prop:support-certificate} gives exactly \eqref{eq:cell-functional}.
\end{proof}

The phase formulation \eqref{eq:cell-functional} is used throughout; it requires no inverse convention on possible plateaus of $m_C$.

\begin{remark}[Reading the log-polar panel]
\label{rem:log-polar}
The canonical square in \cref{fig:overview} is a fundamental domain of the log-polar quotient, not an angular-monotonicity hypothesis. A general homothetic cell may reverse angle, cross the same vertical line several times, and have several radii at the same angle. The curves drawn as graphs are illustrative; neither \cref{thm:cell} nor \eqref{eq:cell-functional} assumes a single-valued representation $\log_Q r=y(\theta)$.
\end{remark}

\begin{figure}[H]
\centering
\includegraphics[width=\textwidth]{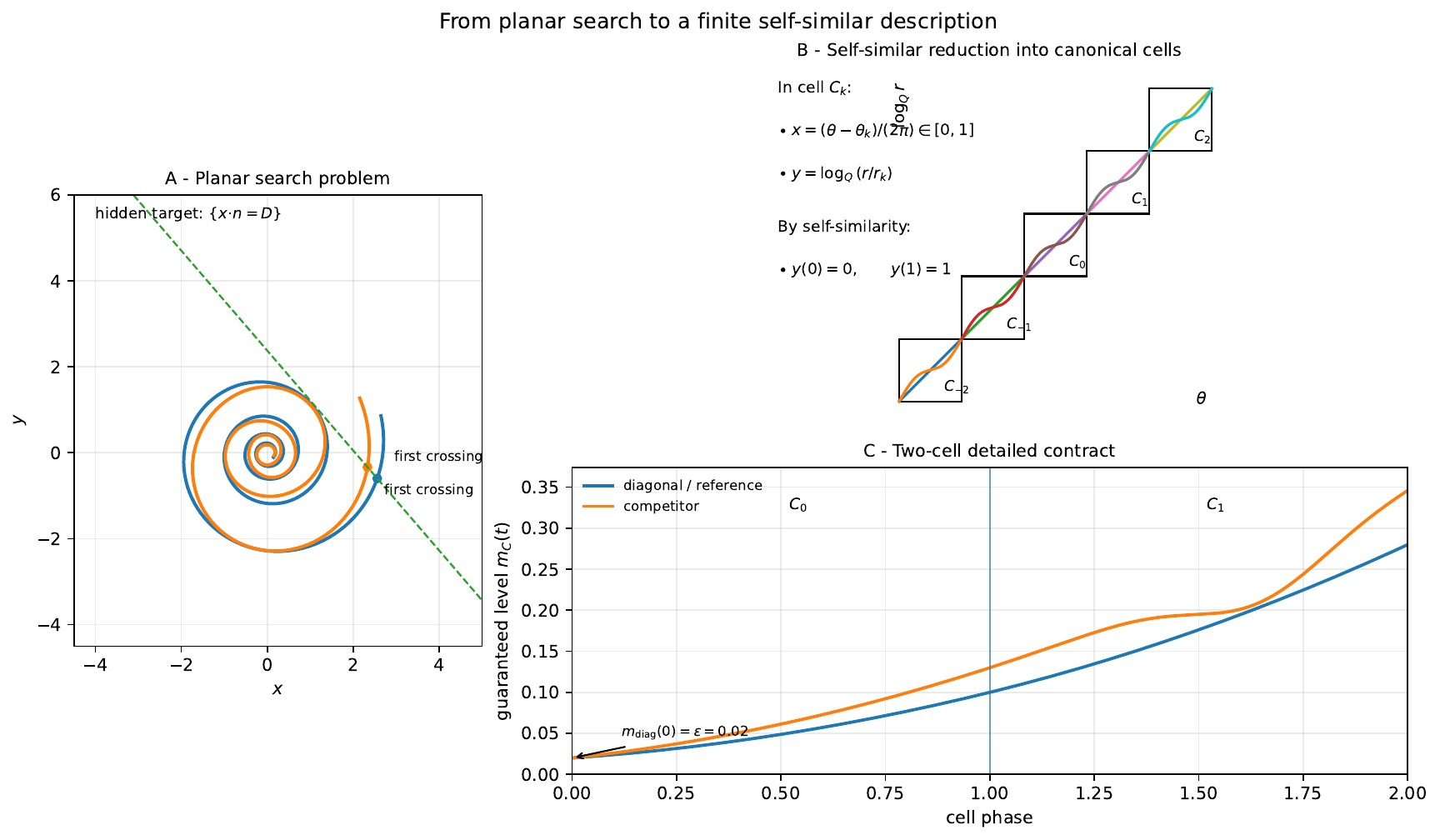}
\caption{From planar search to a finite self-similar description. Panel A compares two illustrative paths against the same hidden target line. Panel B displays five consecutive canonical domains in the log-polar quotient $(\theta,\log_Q r)$;  Panel C plots the guaranteed level $m_C(t)$ over two homologous cells. The plate is generated entirely by code.}
\label{fig:overview}
\end{figure}

\section{From the continuum to finite polygonal cells}
\label{sec:polygonalization}
This section removes the final geometric infinitude. It gives, in turn, an exact evaluator at fixed vertex count, compact parameter bounds, and a conservative polygonal approximation.

\subsection{Exact evaluator}
Consider a polygonal cell
\begin{equation}
p_0,p_1,\ldots,p_N\in\R^2,
\qquad
p_N=Qp_0,
\qquad
Q>1,
\label{eq:polygonal-cell}
\end{equation}
and set
\[
\lambda:=Q^{-1},
\qquad
P:=\conv\{p_0,\ldots,p_N\},
\qquad
\ell_e:=\norm{p_e-p_{e-1}},
\qquad
\ell:=\sum_{e=1}^N\ell_e.
\]
By homogeneity, normalize
\begin{equation}
\inrad_0(P)=1,
\label{eq:inradius-normalization}
\end{equation}
which implies $0\in\operatorname{int}P$. During edge $e$, at phase $\alpha\in[0,1]$, define
\begin{align}
z_{e,\alpha}&:=p_{e-1}+\alpha(p_e-p_{e-1}),
\label{eq:z-edge}\\
K_{e,\alpha}&:=\conv\bigl(\lambda P,p_0,\ldots,p_{e-1},z_{e,\alpha}\bigr),
\label{eq:K-edge}\\
L_{e,\alpha}&:=\frac{\lambda}{1-\lambda}\ell+
\sum_{i<e}\ell_i+\alpha\ell_e.
\label{eq:L-edge}
\end{align}

\begin{theorem}[Exact polygonal evaluator]
\label{thm:polygonal-evaluator}
The ratio of the repetition of \eqref{eq:polygonal-cell} is
\begin{equation}
\CR(p,Q)=
\max_{1\le e\le N}\sup_{0\le\alpha\le1}
\frac{L_{e,\alpha}}{\inrad_0(K_{e,\alpha})}.
\label{eq:polygonal-evaluator}
\end{equation}
\end{theorem}

\begin{proof}
Before the current cell, the historical convex hull is $\lambda P$ and the historical length is $\lambda\ell/(1-\lambda)$. During edge $e$, the history adds exactly the already visited vertices and $z_{e,\alpha}$. Formula \eqref{eq:polygonal-evaluator} follows from \cref{prop:support-certificate}.
\end{proof}

\subsection{Compactness at fixed complexity}
\begin{proposition}[Compact bounds]
\label{prop:compact-bounds}
Fix $N$ and $U>2$. Every polygonal cell with $N$ nonzero edges satisfying \eqref{eq:inradius-normalization} and $\CR(p,Q)\le U$ obeys
\begin{equation}
2\le\ell\le U,
\qquad
\max_i\norm{p_i}\le U-1,
\qquad
\frac{1}{U^N-1}\le\lambda\le1-\frac{2}{U}.
\label{eq:compact-bounds}
\end{equation}
The corresponding sublevel is therefore compact after zero-length subdivisions are allowed, and its infimum is attained.
\end{proposition}

\begin{proof}
At the beginning of the cell, the past has inradius $\lambda$ and cost $\lambda\ell/(1-\lambda)$, hence $\ell/(1-\lambda)\le U$. Since $B_2\subset P$, the cell joins two points separated by at least $2$, so $\ell\ge2$ and $\lambda\le1-2/U$. If $R:=\max_i\norm{p_i}$, a direction opposite a farthest vertex provides a point of the hull with support at most $-1$; thus $\ell\ge R+1$ and $R\le U-1$.

For the lower bound on $\lambda$, write $S_e:=\ell_1+\cdots+\ell_e$ and choose, at the end of edge $e$, a unit direction $u$ opposite to $p_e-p_{e-1}$. Since $p_N=Qp_0$ and $\lambda=Q^{-1}$, one has $p_0=\lambda p_N\in\lambda P$. For every previously visited vertex $p_j$, $j<e$,
\[
u\cdot p_j\le u\cdot p_0+\norm{p_j-p_0}
\le h_{\lambda P}(u)+S_{e-1}
\le\lambda R+S_{e-1}.
\]
The new point does not increase support in direction $u$. Hence the available support at the end of edge $e$ is at most $\lambda R+S_{e-1}$, and the ratio constraint gives
\[
S_e\le U(\lambda R+S_{e-1}).
\]
Iteration yields
\[
S_e\le\lambda R U\frac{U^e-1}{U-1}.
\]
At $e=N$, with $S_N=\ell$, $R\le\ell-1$, and $\ell\le U$,
\[
\lambda\ge
\frac{\ell(U-1)}{RU(U^N-1)}
\ge
\frac{\ell(U-1)}{(\ell-1)U(U^N-1)}
\ge
\frac{1}{U^N-1},
\]
where the last inequality uses $\ell/(\ell-1)\ge U/(U-1)$ for $2\le\ell\le U$.
This proves \eqref{eq:compact-bounds}.
\end{proof}

\subsection{One-sided polygonalization}
\begin{theorem}[Conservative polygonal approximation]
\label{thm:polygonalization}
Let $C_0$ be a rectifiable cell such that $m_C(t)\ge1$ at every phase and $L(C_0)\le L_{\max}$. Partition it into arcs of length at most
\[
h:=\frac{L_{\max}}{N}<1
\]
and replace each arc by its chord while preserving the cutting points. The resulting polygonal cell $C_N$ satisfies
\begin{equation}
\cellratio_Q(C_N)\le\frac{\cellratio_Q(C_0)}{1-h}.
\label{eq:polygonalization-bound}
\end{equation}
In particular, if $\cellratio_Q(C_0)\le C_{\mathrm{ref}}$, then
\[
\cellratio_Q(C_N)
\le \cellratio_Q(C_0)+\frac{C_{\mathrm{ref}}h}{1-h}.
\]
\end{theorem}

\begin{quote}\small
\textit{Interpretation.} Each short arc is replaced by its chord, but a point on that chord is compared with the end of the corresponding original arc, when the entire arc is already present in the historical hull. This one-sided comparison prevents artificial underestimation of the online cost.
\end{quote}

\begin{proof}
Parametrize the original cell by arc length and choose
\[
0=s_0<s_1<\cdots<s_N=L,
\qquad
s_i-s_{i-1}\le h,
\qquad
p_i:=\gamma(s_i).
\]
Let $z$ be a point on the $i$th polygonal chord. Compare it with original time $s_i$, when the whole corresponding arc is already available. Every point of an earlier arc lies within distance $h$ of a visited cutting point, and every point of the current arc lies within distance $h$ of $p_{i-1}$. Therefore
\begin{equation}
\conv\gamma([0,s_i])
\subset
\conv(p_0,\ldots,p_{i-1},z)+hB_2.
\label{eq:hull-neighborhood}
\end{equation}
The same argument for the terminal cell gives $P\subset P_N+hB_2$, hence $Q^{-1}P\subset Q^{-1}P_N+hB_2$. Thus the original historical hull at time $s_i$ lies in the $h$-neighborhood of the polygonal historical hull at $z$, and
\begin{equation}
m_N(z)\ge m(s_i)-h.
\label{eq:inradius-loss}
\end{equation}

Define the total historical costs
\[
A(s):=\frac{L}{Q-1}+s,
\qquad
A_N(z):=\frac{L_N}{Q-1}+t_N(z),
\]
where $L_N$ is the polygonal cell length and $t_N(z)$ is polygonal arc length from $p_0$ to $z$.
Chords do not lengthen any arc, so $A_N(z)\le A(s_i)$. Since $m(s_i)\ge1$,
\[
\frac{A_N(z)}{m_N(z)}
\le
\frac{A(s_i)}{m(s_i)-h}
\le
\frac{\cellratio_Q(C_0)}{1-h}.
\]
Taking the supremum over all edges proves \eqref{eq:polygonalization-bound}.
\end{proof}

\section{Semialgebraic decision at fixed complexity}
Fix $N$ and a rational threshold $c>0$. The variables are $Q$, $\lambda$, the vertex coordinates, the edge lengths, the phase $\alpha$, and a unit direction $u$. The relations
\[
Q\lambda=1,
\qquad
p_N=Qp_0,
\qquad
d_e^2=\norm{p_e-p_{e-1}}^2,
\qquad
d_e>0,
\qquad
\norm{u}^2=1
\]
are polynomial. The support of $K_{e,\alpha}$ is the maximum of a finite list of scalar products, so the inequality
\[
h_{K_{e,\alpha}}(u)\ge\frac{L_{e,\alpha}}{c}
\]
is a finite disjunction of polynomial inequalities. The gauge $\inrad_0(P)=1$ is also first-order. An explicit encoding is given in Appendix~\ref{app:semialgebraic}.

\begin{proposition}[Decidability at fixed $N$]
\label{prop:qe}
For fixed $N$ and $c\in\mathbb{Q}_{>0}$, the existence of a polygonal cell with exactly $N$ nonzero edges, normalized by $\inrad_0(P)=1$ and satisfying $\CR(p,Q)\le c$, is expressible by a first-order formula over the reals with rational coefficients. It is therefore decidable by quantifier elimination \cite{BasuPollackRoy}.
\end{proposition}

\begin{quote}\small
\textit{Interpretation.} At fixed vertex count, all remaining unknowns are real coordinates. Lengths, finite support maxima, the inradius gauge, and the synchronized ratio constraint are polynomial equalities, inequalities, and finite disjunctions. The decision procedure is qualitative and expensive; its role is termination, not numerical efficiency.
\end{quote}

\section{Effective computability of the planar Shoreline value}
\subsection{Statement}
\begin{theorem}[Effective computability of the optimal value]
\label{thm:computability}
For every rational $\eps>0$, there is a terminating algorithm that produces rationals $L_\eps\le U_\eps$ such that
\begin{equation}
L_\eps\le C_2^*\le U_\eps,
\qquad
U_\eps-L_\eps\le\eps.
\label{eq:computability-interval}
\end{equation}
In particular, $C_2^*$ is a computable real in the sense of computable analysis \cite{Weihrauch}.
\end{theorem}

\begin{quote}\small
\textit{Interpretation.} A bounded nearly optimal cell comes from \cref{thm:effective-return}; \cref{thm:polygonalization} makes its combinatorial complexity computable; \cref{prop:qe} decides rational thresholds; rational bisection then produces \eqref{eq:computability-interval}.
\end{quote}

\subsection{A nearly optimal cell of computable size}
Consider the polygonal cell
\begin{equation}
(1,0)\longrightarrow(0,1)\longrightarrow(-1,0)\longrightarrow(0,-1)\longrightarrow(2,0)
\label{eq:explicit-cell}
\end{equation}
with factor $Q=2$. Its length is
\[
L_0=3\sqrt2+\sqrt5.
\]
Its terminal hull contains the diamond with vertices $(\pm1,0)$ and $(0,\pm1)$, whose inradius is $1/\sqrt2$. The preceding copy, contracted by $2$, therefore supplies at every phase a guaranteed level at least $1/(2\sqrt2)$. Historical cost is at most the total length of the past plus the current cell, namely $2L_0$. Hence
\begin{equation}
C_2^*\le\cellratio_Q(C_0)
\le4\sqrt2\,L_0
=24+4\sqrt{10}<37.
\label{eq:crude-upper}
\end{equation}
Fix $C=38$ and replace $\eps$ by $\bar\eps:=\min\{\eps,1\}$. By definition of the infimum, there is a path $\gamma$ such that
\[
\CR(\gamma)<C_2^*+\frac{\bar\eps}{8}<C.
\]
Apply \cref{thm:effective-return} with $D=2$, $U=C$, $\zeta=\bar\eps/8$, and $\lambda=2$. We obtain a cell $C_0$ with
\begin{equation}
\cellratio_Q(C_0)<C_2^*+\frac{\bar\eps}{4},
\qquad
L(C_0)\le L_{\max}\left(2,C,\frac{\bar\eps}{8},2\right).
\label{eq:almost-optimal-cell}
\end{equation}
After a homogeneous normalization, the preceding copy supplies a guaranteed level at least $1$ throughout the cell.

\subsection{Computable polygonal complexity}
Set
\begin{equation}
h_\eps:=\min\left\{\frac12,\frac{\bar\eps}{8C}\right\},
\qquad
N(\eps):=\left\lceil\frac{L_{\max}}{h_\eps}\right\rceil.
\label{eq:N-eps}
\end{equation}
By \cref{thm:polygonalization}, there is a polygonal cell with at most $N(\eps)$ edges and ratio less than $C_2^*+\bar\eps/2$.

Let $W_N$ be the infimum of the ratio over all planar polygonal cells with at most $N=N(\eps)$ nonzero edges, free factor $Q>1$, and gauge $\inrad_0(P)=1$. Every such cell is an admissible strategy, so
\begin{equation}
C_2^*\le W_N\le C_2^*+\frac{\bar\eps}{2}.
\label{eq:WN-sandwich}
\end{equation}

\subsection{Exact oracle and rational bisection}
For $1\le n\le N$, let $\mathrm{FEAS}_n(c)$ assert the existence of a cell with exactly $n$ nonzero edges and ratio at most $c$, and set
\begin{equation}
\mathrm{FEAS}_{\le N}(c):=\bigvee_{n=1}^N\mathrm{FEAS}_n(c).
\label{eq:feas-union}
\end{equation}
By \cref{prop:qe}, this assertion is decidable for every rational $c$.

To justify attainment without confusing ``exactly $n$ nonzero edges'' with a closed parameter set, first allow zero-length subdivisions in an $N$-edge representation. The bounds of \cref{prop:compact-bounds} place the relevant sublevel in a compact box, and the exact evaluator is continuous there. A minimizer therefore exists. Deleting its zero-length edges leaves the same geometric path and the same ratio, with some number $n\le N$ of nonzero edges. Consequently the disjunction \eqref{eq:feas-union} detects the attained minimum.
Thus
\[
\mathrm{FEAS}_{\le N}(c)\quad\Longleftrightarrow\quad W_N\le c.
\]
Rational bisection on $[0,C]$ gives, after finitely many steps, rationals $a\le W_N\le b$ with $b-a\le\bar\eps/2$. By \eqref{eq:WN-sandwich},
\[
a-\frac{\bar\eps}{2}\le C_2^*\le b.
\]
The choice
\[
L_\eps:=\max\left\{0,a-\frac{\bar\eps}{2}\right\},
\qquad
U_\eps:=b
\]
gives \eqref{eq:computability-interval} and completes the proof of \cref{thm:computability}.

\begin{corollary}[Synthesis of an $\eps$-optimal strategy]
\label{cor:synthesis}
A constructive quantifier-elimination procedure can return a feasibility witness. Hence, in principle, one can produce an algebraic factor $Q$ and a polygonal cell with algebraic coordinates whose ratio is at most $C_2^*+\eps$.
\end{corollary}

\begin{corollary}[Certification of every strict gap]
\label{cor:strict-gap}
For every rational $c<C_2^*$, the following procedure terminates: for $n=1,2,\ldots$, apply \cref{thm:computability} with $\eps=2^{-n}$ and stop as soon as $L_{2^{-n}}>c$.
\end{corollary}

\begin{remark}
The covering constants, then $L_{\max}$ and $N(\eps)$, are enormous, and quantifier elimination is prohibitive. \Cref{thm:computability} gives a terminating procedure at arbitrary precision, not a numerically competitive algorithm.
\end{remark}

\section{Scope of the main chain}
Part~I establishes four distinct facts. First, the infimum over all admissible histories equals the infimum over repeated homothetic cells. Second, at fixed precision, a nearly optimal cell can be chosen with computably bounded factor and normalized length. Third, this cell can be replaced by a polygon of computable complexity without artificially underestimating the true online cost. Finally, fixed-complexity polygonal optimization is decidable at a prescribed threshold, making $C_2^*$ computable.

The conclusion is qualitative rather than practical: the covering bounds and the complexity of quantifier elimination are extremely large. It also does not classify extremal cells. The remaining structural problem is to determine whether the computable value equals $\spc$ or whether a nonstationary cell achieves a strictly smaller ratio.

Part~II collects exact identities not required by any step above. They shorten memory, describe a Bellman transition, provide exclusion criteria, and explain why a logarithmic spiral appears when a normalized state becomes stationary up to rotation.

\part{A toolbox for Shoreline search}

\section{Auxiliary results for computation and rigidity}
The results in this part are not used in the proof of \cref{thm:computability}. Each subsection gives the intuition, the useful statement, and its role; complete proofs are collected in Appendices~\ref{app:toolbox-proofs} and~\ref{app:filters}.

\subsection{Homologous recutting}
A critical internal phase may be chosen as a new beginning of the cell: the prefix moved to the end is dilated by $Q$, so the cell remains homologous to itself. For $s\in[0,L]$, set
\begin{equation}
C_0^{(s)}:=C_0[s,L]\star Q C_0[0,s],
\qquad
L(C_0^{(s)})=L+(Q-1)s.
\label{eq:recut-cell}
\end{equation}

\begin{theorem}[Exact recutting]
\label{thm:recut}
For every phase $s$ with $m_C(s)>0$,
\begin{equation}
\inrad_0\bigl(\conv C_0^{(s)}\bigr)=Qm_C(s),
\label{eq:recut-inradius}
\end{equation}
and
\begin{equation}
\cellratio_Q(C_0)=
\frac{Q}{Q-1}
\sup_{0\le s\le L}
\frac{L(C_0^{(s)})}{\inrad_0(\conv C_0^{(s)})}.
\label{eq:recut-functional}
\end{equation}
\end{theorem}

A recut turns a bad prefix into a terminal certificate. It is particularly useful for pruning at fixed $Q$ and for comparing phases without losing the true historical cost. The proof appears in Appendix~\ref{app:recut-proof}.

\subsection{Sliding memory}
Under a competitive bound $C$, every point visited before $t/C$ already lies in the disk guaranteed at time $t$ and can no longer enlarge the convex hull.

\begin{theorem}[Exact sliding memory]
\label{thm:sliding}
If $\CR(\gamma)\le C$, then for every $t>0$,
\begin{equation}
K_t=\conv\gamma([t/C,t]).
\label{eq:sliding-memory}
\end{equation}
In logarithmic time, useful memory has the exact horizon $\log C$.
\end{theorem}

Thus the state depends on a finite multiplicative window rather than on the whole past. This makes a dynamic implementation conceivable. The proof appears in Appendix~\ref{app:sliding-proof}.

\subsection{Bellman transition, predecessors, and deadlines}
In logarithmic time, the old past contracts while a new path segment is added. The support function places these two contributions in competition through a maximum.

Fix a candidate ratio $C$, set $d_C:=1/C$ and $t=e^\tau$, and define
\begin{equation}
p(\tau):=\frac{\gamma(t)}{t},
\qquad
H_\tau:=\frac{K_t}{t},
\qquad
h_\tau:=h_{H_\tau}.
\label{eq:normalized-bellman-state}
\end{equation}

\begin{proposition}[Exact normalized state]
\label{prop:normalized-state}
For every finite threshold $C$, with $d_C=1/C$,
\begin{equation}
\CR(\gamma)\le C
\quad\Longleftrightarrow\quad
\min_{u\in\Sph^1}h_\tau(u)\ge d_C
\quad\text{for all }\tau.
\label{eq:viability-state}
\end{equation}
Every viable state satisfies
\begin{equation}
d_CB_2\subset H_\tau\subset B_2,
\qquad
p(\tau)\in H_\tau.
\label{eq:state-inclusions}
\end{equation}
If $v(\tau):=\gamma'(e^\tau)$ at differentiability times, then
\begin{equation}
p'(\tau)=v(\tau)-p(\tau).
\label{eq:p-dynamics}
\end{equation}
\end{proposition}

For a logarithmic step $\Delta>0$, put $a=e^{-\Delta}$ and $\delta=1-a$, and represent the new segment by a unit-speed curve $c:[0,\delta]\to\R^2$ with $c(0)=ap$.

\begin{theorem}[Exact Bellman step]
\label{thm:bellman-step}
The complete transition is
\begin{equation}
p^+=c(\delta),
\qquad
H^+=\conv\bigl(aH\cup c([0,\delta])\bigr),
\qquad
h^+(u)=\max\{ah(u),h_c(u)\}.
\label{eq:bellman-step}
\end{equation}
\end{theorem}

\begin{proposition}[Backward constraint and deadline]
\label{prop:deadline}
If $h_c(u)<h_{H^+}(u)$, then terminal support in direction $u$ must be inherited from the past:
\begin{equation}
h_{aH}(u)=h_{H^+}(u).
\label{eq:backward-support}
\end{equation}
Every state viable at ratio $C$ also satisfies
\begin{equation}
\ip{u}{p}\ge1-(C-1)h(u).
\label{eq:deadline}
\end{equation}
At the competitive floor this becomes an identity:
\begin{equation}
h(u)=d_C
\quad\Longrightarrow\quad
\ip{u}{p}=d_C.
\label{eq:floor-contact}
\end{equation}
\end{proposition}

The forward step gives an autonomous dynamics. The backward constraint describes possible predecessors, while the deadline immediately eliminates a state that cannot refresh a support direction before it becomes critical. Proofs appear in Appendices~\ref{app:bellman-proof} and~\ref{app:deadline-proof}.

\subsection{Stored support and relative equilibrium}
\label{sec:relative-equilibria}
Support above the minimum is prepaid service carried by max-memory. A state that preserves its shape and merely rotates in logarithmic time becomes a spiral in physical space.

For each direction $u$, write $Z_t(u):=h_{K_t}(u)$ and define
\begin{equation}
b_C(t):=m(t)-\frac{t}{C},
\qquad
\sigma_t(u):=Z_t(u)-m(t).
\label{eq:stored-support}
\end{equation}
Then
\begin{equation}
Z_t(u)-\frac{t}{C}=b_C(t)+\sigma_t(u).
\label{eq:reserve-decomposition}
\end{equation}
The term $b_C$ is a common reserve and $\sigma_t(u)$ a directional reserve. When a point $x$ is added to a convex body $K$, the fresh support created in direction $u$ is exactly
\begin{equation}
\Delta Z(u)=\bigl[\ip{u}{x}-h_K(u)\bigr]_+.
\label{eq:fresh-support}
\end{equation}
At the floor $h(u)=d_C$, identity \eqref{eq:floor-contact} forces the current point itself to carry every critical direction.

\begin{proposition}[Exact price of fresh support]
\label{prop:fresh-support-price}
Let $A,X,B\in\R^2$, let $S=[A,B]$, $T=\conv\{A,X,B\}$, and
\[
\eps_X:=\norm{X-A}+\norm{B-X}-\norm{B-A}\ge0.
\]
For $u_\theta=(\cos\theta,\sin\theta)$, set
\[
f_X(\theta):=h_T(u_\theta)-h_S(u_\theta)\ge0.
\]
Then
\begin{equation}
\int_0^{2\pi}f_X(\theta)\,d\theta=\eps_X.
\label{eq:fresh-price}
\end{equation}
More generally, after masking by an old support $h_K$, if
\[
Z_X:=\max\{h_K,h_T\},
\qquad
Z_0:=\max\{h_K,h_S\},
\]
then
\begin{equation}
0\le Z_X-Z_0\le f_X,
\qquad
\int_0^{2\pi}(Z_X-Z_0)\,d\theta\le\eps_X.
\label{eq:masked-price}
\end{equation}
\end{proposition}

\begin{remark}[Clock credit]
\label{rem:clock-credit}
Suppose a block is replaced by a block with the same endpoints, shorter by $\eps>0$, viable during its execution, and whose terminal support loses at most $\eps/C$ in every direction. Then the same continuation remains viable. The length saving advances the clock by $\eps$ and lowers the future barrier by exactly $\eps/C$.
\end{remark}

Further budget and stock-renewal identities are given in Appendix~\ref{app:support-budget}.

\begin{theorem}[Relative equilibrium and logarithmic spiral]
\label{thm:relative-equilibrium}
Suppose there exist $\omega\ne0$ and a state $(p_*,H_*)$ such that
\begin{equation}
p(\tau)=R_{\omega\tau}p_*,
\qquad
H_\tau=R_{\omega\tau}H_*.
\label{eq:relative-equilibrium}
\end{equation}
Then
\begin{equation}
\gamma(t)=tR_{\omega\log t}p_*,
\qquad
r(\theta)=r_0e^{\kappa\theta},
\qquad
\kappa=\omega^{-1}.
\label{eq:log-spiral}
\end{equation}
Within the stationary family, tangency between the current point and an old provider delayed by an angle $\tau\in(\pi,2\pi)$ gives
\begin{align}
e^{-\kappa\tau}&=\cos\tau-\kappa\sin\tau,
\label{eq:spiral-tangency}\\
C_{\mathrm{rel}}(\kappa)&=
\frac{1+\kappa^2}{\kappa}e^{\kappa\tau(\kappa)},
\label{eq:spiral-ratio}\\
\tan\tau&=\frac{\kappa^2\tau}{1-\kappa\tau}.
\label{eq:spiral-stationarity}
\end{align}
\end{theorem}

Tangency is a local condition. To identify the actual support face, one must also verify that the candidate line dominates the entire past spiral and that the perpendicular foot lies on the active chord. \Cref{lem:active-chord} proves both facts for the branch $\tau\in(\pi,2\pi)$ used here. The active solution is
\begin{equation}
\kappa_*\simeq0.2124695594156479,
\qquad
\tau_*\simeq4.8588823628256375,
\qquad
Q_*=e^{2\pi\kappa_*}\simeq3.79994128,
\label{eq:spiral-numerics}
\end{equation}
and $C_{\mathrm{rel}}(\kappa_*)=\spc\simeq13.81113517946$ \cite{FinchSpiral}. These equations evaluate the best stationary logarithmic-spiral family; they are not used in the proof of computability and do not establish global optimality.

\subsection{Computational filters and local diagnostics}
\label{sec:filters}
The constraints in \cref{tab:filters} are one-sided: violating one certifies impossibility, while satisfying all of them does not certify historical realizability.

\begin{proposition}[Terminal bound at fixed scale factor]
\label{prop:terminal-bound}
Every planar homothetic cell with factor $Q>1$, length $L$, and terminal inradius $M$ satisfies
\begin{equation}
\frac{L}{M}\ge
\Lambda(Q):=2\pi-2\arcsin\left(\frac{Q-1}{2Q}\right),
\label{eq:Lambda}
\end{equation}
therefore
\begin{equation}
\cellratio_Q(C_0)\ge
\Phi(Q):=\frac{Q}{Q-1}\Lambda(Q).
\label{eq:Phi}
\end{equation}
The function $\Phi$ is strictly decreasing. The equation $\Phi(Q_1)=\spc$ has the unique solution
\begin{equation}
Q_1\simeq1.7361155691.
\label{eq:Q1}
\end{equation}
Thus every cell that might strictly beat the spiral reference value must satisfy $Q>Q_1$.
\end{proposition}

\begin{proposition}[Anisotropy cage]
\label{prop:anisotropy-cage}
Assume $\cellratio_Q(C_0)<\spc$ and consider a recut cell of length $\ell$, terminal hull $K$, inradius $M$, and support $h$. Set
\[
z(\theta):=\frac{h(u_\theta)}{M}\ge1,
\qquad
E:=2\ell-\Per(K)\ge0,
\]
and
\begin{equation}
A_{\mathrm{sp}}(Q):=\spc\frac{Q-1}{Q},
\qquad
B_{\mathrm{sp}}(Q):=2\bigl(A_{\mathrm{sp}}(Q)-\pi\bigr).
\label{eq:Asp-Bsp}
\end{equation}
Then
\begin{equation}
\int_0^{2\pi}(z(\theta)-1)\,d\theta+\frac{E}{M}
<B_{\mathrm{sp}}(Q).
\label{eq:anisotropy-budget}
\end{equation}
In particular, for every $a>0$,
\begin{equation}
\abs{\{\theta:z(\theta)\ge e^a\}}
<\frac{B_{\mathrm{sp}}(Q)}{e^a-1},
\label{eq:anisotropy-width}
\end{equation}
and
\begin{equation}
\frac{h_{\max}}{M}
<\min\left\{
A_{\mathrm{sp}}(Q)-1,
1+\sqrt{B_{\mathrm{sp}}(Q)\bigl(A_{\mathrm{sp}}(Q)-1\bigr)}
\right\}.
\label{eq:anisotropy-cap}
\end{equation}
\end{proposition}

\begin{table}[ht]
\centering
\small
\caption{Exact filters and their practical role.}
\label{tab:filters}
\begin{tabularx}{\textwidth}{@{}p{3.2cm}X X@{}}
\toprule
Result & Statement & Use \\
\midrule
Perimeter bound & Every finite-ratio planar strategy satisfies $C\ge1+\pi$. & Elementary consistency check. \\
Fixed-$Q$ bound & $\cellratio_Q(C)\ge\Phi(Q)$; in particular $Q\le Q_1$ cannot beat $\spc$. & Removes an entire fiber of scale factors. \\
Anisotropy cage & Integrated budget \eqref{eq:anisotropy-budget}, width \eqref{eq:anisotropy-width}, and caps \eqref{eq:anisotropy-cap}. & Restricts the height and angular width of support peaks. \\
Recut barrier & After recutting at a phase maximizing $\ell_s/M_s$, $M(t)\ge M(0)(1+(Q-1)t/L)$. & Pruning at fixed $Q$. \\
Directional renewal & $(1-Q^{-1})h(\theta)\le\int_0^L[\cos(\theta-\varphi(s))]_+\,ds$. & Necessary direction-by-direction test. \\
Logarithmic budget & $\int_0^L ds/m_C(s)\ge\pi\log Q$. & Global bound independent of contact order. \\
Two-contact chamber & Exact formulas for rotation of the active face and evolution of the bottleneck. & Local diagnostic for refinement schemes. \\

\bottomrule
\end{tabularx}
\end{table}

Proofs are given in Appendix~\ref{app:filters}. No local numerical constant without an autonomous certificate is used in the main chain.

\section{Conclusion}
The separation into two parts makes the role of every result explicit. The main chain turns an infinite history into a finite decision problem and proves computability of $C_2^*$. The toolbox shortens memory, gives an exact dynamics, constructs exclusion criteria, and relates relative equilibria to logarithmic spirals.

The next challenge is not to prove that the value exists as a computable object, but to make certificates small enough to yield new numerical bounds and to determine whether the spiral regime actually realizes that value.

\appendix
\section{Effective covering of normalized states}
\label{app:covering}
Under the bounds \eqref{eq:state-bounds}, fix $\delta>0$. Use
\begin{itemize}
\item a rational partition of $[1,U]$ into intervals of diameter $\delta/6$;
\item a rational net of the ball $B(0,U)\subset\R^D$ with covering radius $\delta/12$;
\item a rational net $\mathcal N\subset\Sph^{D-1}$ of mesh $\delta/(24U)$;
\item quantization of the values $H(v)$, $v\in\mathcal N$, with step $\delta/12$.
\end{itemize}
If two states have the same code, their length and position components differ by less than $\delta/6$. For the profiles, choose for each $u\in\Sph^{D-1}$ a point $v\in\mathcal N$ with $\norm{u-v}\le\delta/(24U)$. The Lipschitz estimate in \eqref{eq:state-bounds} gives
\begin{align*}
\abs{H(u)-H'(u)}
&\le\abs{H(u)-H(v)}+\abs{H(v)-H'(v)}+\abs{H'(v)-H'(u)}\\
&\le U\frac{\delta}{24U}+\frac{\delta}{12}+U\frac{\delta}{24U}
=\frac{\delta}{6}.
\end{align*}
Thus $D_\Sigma(\Sigma,\Sigma')<\delta/2<\delta$. If $K_\ell$ and $K_x$ are the numbers of cells in the first two grids, $M_S:=\abs{\mathcal N}$, and $J$ is the number of quantization levels, one may take
\[
M(D,U,\delta)\le K_\ell K_x J^{M_S}.
\]
No optimization of this bound is needed; its only role is to provide an explicit integer for the pigeonhole argument in \cref{thm:effective-return}.

\section{Proofs for the Shoreline toolbox}
\label{app:toolbox-proofs}
\subsection{Proof of homologous recutting}
\label{app:recut-proof}
\begin{proof}[Proof of \cref{thm:recut}]
Write $A_s$ for the support of $C_0[0,s]$, $B_s$ for that of $C_0[s,L]$, and $H=A_s\vmax B_s$ for the terminal support. The recut cell has support $B_s\vmax QA_s$. On the other hand,
\[
Q\max\{Q^{-1}H,A_s\}=H\vmax QA_s.
\]
If $A_s(u)\ge0$, the two maxima agree immediately. If $A_s(u)<0$, positivity of the terminal inradius implies $H(u)>0$, hence $B_s(u)=H(u)>QA_s(u)$. This proves \eqref{eq:recut-inradius}; formula \eqref{eq:recut-functional} follows from \eqref{eq:cell-functional} and \eqref{eq:recut-cell}.
\end{proof}

\subsection{Proof of sliding memory}
\label{app:sliding-proof}
\begin{proof}[Proof of \cref{thm:sliding}]
The competitive bound gives $K_t\supset(t/C)B_2$. Every point visited before $t/C$ has norm at most $t/C$ and therefore belongs to this disk. In every direction, support of level $t/C$ is attained within the window $[t/C,t]$; hence the recent convex hull contains the disk and absorbs all earlier history.
\end{proof}

\subsection{Proof of the Bellman step and the predecessor interval}
\label{app:bellman-proof}
\begin{proof}[Proof of \cref{prop:normalized-state}]
Applying \cref{prop:support-certificate} to $K_t/t$ gives
\[
\CR(\gamma)\le C
\Longleftrightarrow
\inrad_0(H_\tau)\ge d_C
\Longleftrightarrow
\min h_\tau\ge d_C.
\]
Since $\norm{\gamma(s)}\le s\le t$, one has $K_t\subset tB_2$, hence $H_\tau\subset B_2$; the other inclusion follows from the ratio bound, and $p(\tau)\in H_\tau$ is immediate. Finally,
\[
p(\tau)=e^{-\tau}\gamma(e^\tau),
\]
so almost everywhere
\[
p'(\tau)=-p(\tau)+\gamma'(e^\tau)=v(\tau)-p(\tau).
\]
\end{proof}

\begin{proof}[Proof of \cref{thm:bellman-step}]
Let the initial physical time be $t_0$ and the final one $t_1=e^\Delta t_0$. Dividing the old hull by $t_1$ gives $aH$, while the new segment has normalized length $(t_1-t_0)/t_1=\delta$. The support function of the convex hull of a union is the maximum of the two support functions.
\end{proof}

In the notation of \cref{thm:bellman-step}, put
\[
A:=aH,
\qquad
C_c:=\conv c([0,\delta]),
\qquad
H^+:=\conv(A\cup C_c).
\]
Directions for which $h_{C_c}<h_{H^+}$ force $h_A=h_{H^+}$; the remaining directions impose only $h_A\le h_{H^+}$. At the threshold $d_C$, define
\[
E_{\mathrm{old}}:=\Ext(H^+)\setminus C_c,
\qquad
A_{\min}:=\conv\bigl(ad_CB_2\cup\{c(0)\}\cup E_{\mathrm{old}}\bigr),
\qquad
A_{\max}:=H^+\cap aB_2.
\]
Every geometrically viable predecessor satisfies
\[
A_{\min}\subset A\subset A_{\max}.
\]
Conversely, every convex body in this interval reproduces the terminal hull at the geometric level; full viability additionally requires an admissible history realizing that predecessor.

\subsection{Proof of the directional deadline}
\label{app:deadline-proof}
\begin{proof}[Proof of \eqref{eq:deadline}]
Fix $t>0$ and a direction $u$, and put $Z:=Z_t(u)=h_{K_t}(u)$. For $\eps>0$, the line
\[
H_\eps:=\{x:\ip{u}{x}=Z+\eps\}
\]
has not yet been met at time $t$. Starting from $\gamma(t)$, every unit-speed continuation must travel at least
\[
Z+\eps-\ip{u}{\gamma(t)}
\]
before meeting it. Therefore
\[
T_\gamma(Z+\eps)
\ge t+Z+\eps-\ip{u}{\gamma(t)}.
\]
The ratio bound gives $T_\gamma(Z+\eps)\le C(Z+\eps)$, hence
\[
\ip{u}{\gamma(t)}\ge t-(C-1)(Z+\eps).
\]
Letting $\eps\downarrow0$ and dividing by $t$ gives \eqref{eq:deadline}. If $h(u)=1/C$, membership $p\in H$ also gives $\ip{u}{p}\le h(u)=1/C$; combined with \eqref{eq:deadline}, this yields \eqref{eq:floor-contact}.
\end{proof}

\subsection{Support budget and renewal}
\label{app:support-budget}
\subsubsection{Support budget}
Let a planar history have length $L$, convex hull $K$, and inradius $m>0$. Set
\[
\Sigma:=\int_0^{2\pi}\bigl(h_K(u_\theta)-m\bigr)\,d\theta.
\]
Cauchy's formula and total variation of projections give
\[
\Per(K)=2\pi m+\Sigma\le2L.
\]
With $E:=2L-\Per(K)\ge0$,
\begin{equation}
2L=2\pi m+\Sigma+E.
\label{eq:support-budget}
\end{equation}
For $z(u):=h_K(u)/m\ge1$,
\[
\frac{\Sigma}{m}
=\int\log z(u)\,du
 +\int\bigl(z(u)-1-\log z(u)\bigr)\,du.
\]
The first term is compatible with multiplicative transport of stored support; the second is a nonnegative dissipation of anisotropy.

\subsubsection{Price of fresh support and clock credit}
\begin{proof}[Proof of \cref{prop:fresh-support-price}]
Cauchy's formula gives
\[
\int_0^{2\pi}h_T(u_\theta)\,d\theta=\Per(T),
\qquad
\int_0^{2\pi}h_S(u_\theta)\,d\theta=2\norm{B-A}.
\]
The perimeter of the possibly degenerate triangle $T$ is
\[
\norm{X-A}+\norm{B-X}+\norm{B-A}.
\]
The difference of the two integrals is therefore $\eps_X$. For masking, $x\mapsto\max\{a,x\}$ is monotone and $1$-Lipschitz. Since $h_T\ge h_S$, pointwise
\[
0\le\max\{h_K,h_T\}-\max\{h_K,h_S\}\le h_T-h_S.
\]
Integration gives \eqref{eq:masked-price}.
\end{proof}

For the clock credit, write $Z_1$ and $\widetilde Z_1$ for supports at the splice and assume $\widetilde Z_1\ge Z_1-\eps/C$. At every geometric point of the common continuation, the repaired path arrives $\eps$ earlier, so the competitive barrier has dropped by $\eps/C$. Max-memory cannot amplify the initial loss; the continuation remains viable.

\subsubsection{Renewal between completion times}
Let $d_i$ be genuine completion levels, $T_i$ their times, and set
\[
R_i:=\frac{T_i}{d_i},
\qquad
\xi_i:=\log\frac{d_{i+1}}{d_i},
\qquad
c_i:=\frac{T_{i+1}-T_i}{d_{i+1}}.
\]
The identity
\begin{equation}
R_{i+1}=e^{-\xi_i}R_i+c_i
\label{eq:completion-cocycle}
\end{equation}
expresses historical cost as a debt that is discounted and then replenished. For a reference level $C_0$ and $E_i=R_i-C_0$,
\[
E_{i+1}=e^{-\xi_i}E_i+
\bigl[c_i-(1-e^{-\xi_i})C_0\bigr].
\]
Pre-service may move this debt from one phase to the next, but it cannot remove it from the dynamics.

\subsection{Proof of the relative equilibrium formulas}
\begin{proof}[Proof of \cref{thm:relative-equilibrium}]
By definition, $\gamma(t)=tp(\log t)$. Under \eqref{eq:relative-equilibrium}, the radius is proportional to $t$ and the angle equals $\omega\log t+\theta_0$; eliminating $t$ gives \eqref{eq:log-spiral}.

In the planar stationary family, the old critical provider is delayed by an angle $\tau$. Equality of support and tangency give \eqref{eq:spiral-tangency}; the historical quotient gives \eqref{eq:spiral-ratio}; stationarity in $\kappa$ gives \eqref{eq:spiral-stationarity}. These calculations are standard in the analysis of the logarithmic-spiral family \cite{FinchSpiral}.
\end{proof}

\begin{lemma}[Global validity of the active chord]
\label{lem:active-chord}
Let $\kappa>0$ and $\tau\in(\pi,2\pi)$ satisfy
\[
r:=e^{-\kappa\tau}=\cos\tau-\kappa\sin\tau>0.
\]
Set $P_\kappa(s):=e^{-\kappa s}R_{-s}e_1$ for $s\ge0$, $A=P_\kappa(0)$, and $B=P_\kappa(\tau)$. Then the line through $A$ and $B$ is a support line of $\conv P_\kappa([0,\infty))$, and the perpendicular foot from the origin lies on the segment $[A,B]$.
\end{lemma}

\begin{proof}
The equation $r=\cos\tau-\kappa\sin\tau$ is equivalent to
\[
\det(A-B,P_\kappa'(\tau))=0,
\]
so chord $AB$ is tangent to the spiral at $B$. Let $n$ be a unit normal to this line, oriented so that
\[
\mu:=\ip{n}{A}=\ip{n}{B}>0.
\]
The function $f(s):=\ip{n}{P_\kappa(s)}$ is a damped sinusoid $e^{-\kappa s}(a\cos s+b\sin s)$. Tangency gives $f'(\tau)=0$, and the differential identity
\[
f''+2\kappa f'+(1+\kappa^2)f=0
\]
gives $f''(\tau)=-(1+\kappa^2)\mu<0$, so $\tau$ is a local maximum. The zeros of $f'$ are spaced by $\pi$; since $\tau\in(\pi,2\pi)$, the only critical point in $(0,\tau)$ is $\tau-\pi$, which is a minimum. As $f(0)=f(\tau)=\mu$, one has $f(s)\le\mu$ on $[0,\tau]$. For $s>\tau$, subsequent positive maxima occur at $\tau+2m\pi$ and equal $e^{-2m\pi\kappa}\mu<\mu$. Thus $f(s)\le\mu$ for every $s\ge0$, proving the support property.

With $A=(1,0)$ and $B=r(\cos\tau,-\sin\tau)$, the perpendicular foot is $A+\alpha(B-A)$ with
\[
\alpha=\frac{1-r\cos\tau}{1+r^2-2r\cos\tau}.
\]
Clearly $\alpha>0$, and
\[
(1+r^2-2r\cos\tau)-(1-r\cos\tau)
=r(r-\cos\tau)
=-\kappa r\sin\tau>0
\]
because $\tau\in(\pi,2\pi)$. Hence $0<\alpha<1$ and the foot lies on $[A,B]$.
\end{proof}

\section{Proofs of the certified filters}
\label{app:filters}
\subsection{Deadlines and the perimeter bound}
The deadline inequality \eqref{eq:deadline} was proved in Appendix~\ref{app:deadline-proof}. It yields an elementary universal bound.

\begin{proposition}
\label{prop:perimeter-bound}
Every finite-ratio planar strategy satisfies $C\ge1+\pi$.
\end{proposition}

\begin{proof}
In the normalized state, \eqref{eq:deadline} gives
\[
h(u)\ge\frac{1-\ip{u}{p}}{C-1}.
\]
The right-hand side is the support function of a disk of radius $1/(C-1)$ translated by $-p/(C-1)$. The normalized convex hull therefore contains this disk and has perimeter at least $2\pi/(C-1)$. The convex hull of a curve of normalized length $1$ has perimeter at most $2$, so $C\ge1+\pi$.
\end{proof}

\subsection{Recutting and renewal at fixed scale factor}
For a cell $C_0:x\to Qx$, \cref{thm:recut} gives, at phase $s$,
\[
\ell_s=L+(Q-1)s,
\qquad
M_s=Qm_C(s),
\qquad
\cellratio_Q(C_0)=\frac{Q}{Q-1}\sup_s\frac{\ell_s}{M_s}.
\]
Recut at a phase maximizing $\ell_s/M_s$ and call that phase $0$. For every later phase $t$,
\[
\frac{L+(Q-1)t}{M(t)}\le\frac{L}{M(0)},
\]
hence
\begin{equation}
M(t)\ge M(0)\left(1+\frac{Q-1}{L}t\right).
\label{eq:recut-barrier}
\end{equation}

\subsubsection{Terminal bound from variation of all projections}
Normalize $M=1$ and, after rotation, write the endpoints as $A=re_1$ and $B=Qre_1$. For $u_\theta=(\cos\theta,\sin\theta)$, the scalar projection $x_\theta=\ip{u_\theta}{\gamma}$ must visit levels $-1$ and $+1$, because the terminal hull contains the unit disk. With $z=r\abs{\cos\theta}$, its total variation is therefore at least
\begin{equation}
F_Q(z):=
\begin{cases}
4-(Q-1)z,&0\le z\le Q^{-1},\\
2+(Q+1)z,&z\ge Q^{-1}.
\end{cases}
\label{eq:FQ}
\end{equation}
On the other hand,
\[
\int_0^{2\pi}\operatorname{Var}(x_\theta)\,d\theta=4L.
\]
By symmetry,
\[
L\ge\int_0^{\pi/2}F_Q(r\cos\theta)\,d\theta.
\]
Set $R=Qr$ and $\alpha=(Q-1)/Q$. For $R>1$, the right-hand side equals
\[
\pi+2\arcsin\frac1R+2\sqrt{R^2-1}-\alpha R.
\]
Its derivative is $2\sqrt{R^2-1}/R-\alpha$; the unique minimum occurs at
\[
R=\frac{1}{\sqrt{1-\alpha^2/4}}
\]
and equals $2\pi-2\arcsin(\alpha/2)$. The same value dominates the regime $R\le1$. Restoring scale $M$ gives \eqref{eq:Lambda}. Recutting at the initial phase gives
\[
\cellratio_Q(C_0)\ge\frac{Q}{Q-1}\frac{L}{M},
\]
which proves \eqref{eq:Phi}. In the variable $\alpha$, the function
\[
\Phi(\alpha)=\frac{2\pi-2\arcsin(\alpha/2)}{\alpha}
\]
is strictly decreasing, giving uniqueness of $Q_1$.

\subsubsection{Anisotropy cage}
For every recut, Cauchy's formula and \eqref{eq:support-budget} give
\begin{equation}
\frac{\ell}{M}
=\pi+\frac12\left(
\int_0^{2\pi}(z(\theta)-1)\,d\theta+\frac{E}{M}
\right).
\label{eq:anisotropy-identity}
\end{equation}
If $\cellratio_Q(C_0)<\spc$, the recut formula \eqref{eq:recut-functional} imposes $\ell/M<A_{\mathrm{sp}}(Q)$ at every phase, proving \eqref{eq:anisotropy-budget}. Markov's inequality applied to $z-1$ gives \eqref{eq:anisotropy-width}.

Let $R:=\max_{x\in K}\norm{x}$. A curve whose hull contains the disk $MB_2$ must join a point of radius $R$ to a point at distance at least $R+M$ in the opposite direction; hence $\ell\ge R+M$, and $R/M<A_{\mathrm{sp}}(Q)-1$. Since the support function is $R$-Lipschitz in angle, a peak of height $e=z_{\max}-1$ has area at least $e^2/(R/M)$. The cage gives
\[
e^2<B_{\mathrm{sp}}(Q)\frac{R}{M}
<B_{\mathrm{sp}}(Q)\bigl(A_{\mathrm{sp}}(Q)-1\bigr),
\]
which proves the second cap in \eqref{eq:anisotropy-cap}; the first follows from $h_{\max}\le R$.

Let $h(\theta)$ be terminal support and $\varphi(s)$ the direction of unit velocity. The preceding copy initially supplies $Q^{-1}h(\theta)$, so the current cell must renew
\[
\Delta h(\theta)=\left(1-\frac1Q\right)h(\theta).
\]
Support in direction $\theta$ cannot grow faster than $[\cos(\theta-\varphi(s))]_+$, hence
\begin{equation}
\left(1-\frac1Q\right)h(\theta)
\le\int_0^L[\cos(\theta-\varphi(s))]_+\,ds.
\label{eq:directional-renewal}
\end{equation}
For the logarithmic budget, let $H_s(\theta)$ be historical support at phase $s$ and set
\[
r_s(\theta):=\log\frac{H_s(\theta)}{Q^{-1}h(\theta)}.
\]
Then $r_0=0$, $r_L=\log Q$, and because $H_s(\theta)\ge m_C(s)$,
\[
\partial_s r_s(\theta)
\le\frac{[\cos(\theta-\varphi(s))]_+}{m_C(s)}.
\]
Integrating over $\theta$ and using
\[
\int_0^{2\pi}[\cos(\theta-\varphi)]_+\,d\theta=2
\]
gives
\begin{equation}
\int_0^L\frac{ds}{m_C(s)}\ge\pi\log Q.
\label{eq:log-budget}
\end{equation}

\subsection{A regular two-contact chamber}
Assume the bottleneck face is carried by the current point $p$ and an old provider $y$. Let $n$ be its unit normal, $t=Jn$ its tangent, and write
\[
p=mn+s_-t,
\qquad
y=mn+s_+t,
\qquad
\ell=s_+-s_->0.
\]
For a unit control $v$, as long as this chamber remains active, the angular velocity $\Omega$ of the normal and the inradius derivative satisfy
\begin{equation}
\Omega=\frac{n\cdot v}{\ell},
\qquad
\dot m=-m+\frac{s_+}{\ell}(n\cdot v).
\label{eq:two-contact-m}
\end{equation}
If the historical branch has radius of curvature $\rho_+$ at the old contact, then
\begin{equation}
\dot\ell=-\ell-t\cdot v+\frac{\rho_+}{\ell}(n\cdot v).
\label{eq:two-contact-length}
\end{equation}
Indeed, in the normalized frame $p'=v-p$ and $n'=\Omega t$. Differentiating $n\cdot p=m$ gives
\[
\dot m=n\cdot v-m+\Omega s_-.
\]
The same differentiation at the old contact equates the normal velocities of the two contacts; their difference yields $\Omega=(n\cdot v)/\ell$, and then \eqref{eq:two-contact-m}. Differentiating $s_+-s_-$, with the tangential velocity of the old contact on a branch of curvature radius $\rho_+$, gives \eqref{eq:two-contact-length}.

\section{Explicit semialgebraic encoding}
\label{app:semialgebraic}
For fixed $n$, use variables
\[
Q,\lambda,
\quad
p_{i,1},p_{i,2}\ (0\le i\le n),
\quad
d_i\ (1\le i\le n),
\]
with
\[
Q\lambda=1,
\qquad
Q>1,
\qquad
p_n=Qp_0,
\qquad
d_i>0,
\qquad
d_i^2=(p_{i,1}-p_{i-1,1})^2+(p_{i,2}-p_{i-1,2})^2.
\]
For each edge $e$, universally quantify $\alpha\in[0,1]$ and $u=(u_1,u_2)$ with $u_1^2+u_2^2=1$. Put
\[
z_{e,\alpha}=p_{e-1}+\alpha(p_e-p_{e-1})
\]
and
\[
V_{e,\alpha}:=
\{\lambda p_j:0\le j\le n\}
\cup
\{p_0,\ldots,p_{e-1},z_{e,\alpha}\}.
\]
After multiplication by $1-\lambda>0$, the ratio constraint is
\begin{align*}
\forall e,\alpha,u:
&\quad
\bigl(0\le\alpha\le1\wedge\norm{u}^2=1\bigr)\\
&\Longrightarrow
\bigvee_{v\in V_{e,\alpha}}
 c(1-\lambda)u\cdot v
\ge
\lambda\sum_{i=1}^n d_i
+(1-\lambda)\left(\sum_{i<e}d_i+\alpha d_e\right).
\end{align*}
The gauge $\inrad_0(P)=1$ is the conjunction of
\[
\forall u:\norm{u}^2=1
\Longrightarrow
\bigvee_{0\le j\le n}u\cdot p_j\ge1
\]
and
\[
\exists u:\norm{u}^2=1,
\qquad
u\cdot p_j\le1\quad(0\le j\le n).
\]
All expressions are polynomial. Thus $\mathrm{FEAS}_n(c)$ is a first-order formula over the reals, and
\[
\mathrm{FEAS}_{\le N}(c)=\bigvee_{n=1}^N\mathrm{FEAS}_n(c)
\]
is a finite decidable disjunction.

\end{document}